\documentclass[11pt]{article}%
\usepackage{amsfonts}
\usepackage{amsmath}
\usepackage[colorlinks=true,linkcolor=blue,citecolor=red,plainpages=false,pdfpagelabels]%
{hyperref}
\usepackage{amssymb}
\usepackage{fullpage}
\usepackage{graphicx}%
\providecommand{\U}[1]{\protect\rule{.1in}{.1in}}
\newtheorem{theorem}{Theorem}

\newtheorem{corollary}[theorem]{Corollary}

\newtheorem{lemma}[theorem]{Lemma}

\newtheorem{proposition}[theorem]{Proposition}
\newtheorem{remark}[theorem]{Remark}

\newenvironment{proof}[1][Proof]{\noindent\textbf{#1.} }{\ \rule{0.5em}{0.5em}}
\begin{document}

\title{\textbf{Amortization does not enhance the max-Rains information of a quantum channel}}
\author{Mario Berta\thanks{Department of Computing, Imperial College London, London SW7 2AZ, UK, and 
Institute for Quantum Information and Matter, California Institute of
Technology, Pasadena, California 91125, USA} \and Mark M.~Wilde\thanks{Hearne Institute for Theoretical Physics, Department of Physics and Astronomy,
Center for Computation and Technology, Louisiana State University, Baton
Rouge, Louisiana 70803, USA}}
\date{\today}

\maketitle

\begin{abstract}
Given an entanglement measure $E$, the entanglement of a quantum channel is
defined as the largest amount of entanglement $E$ that can be generated from
the channel, if the sender and receiver are not allowed to share a quantum
state before using the channel. The amortized entanglement of a quantum
channel is defined as the largest net amount of entanglement $E$\ that can be
generated from the channel, if the sender and receiver are allowed to share an
arbitrary state before using the channel. Our main technical result is that
amortization does not enhance the entanglement of an arbitrary quantum
channel, when entanglement is quantified by the max-Rains relative entropy. We
prove this statement by employing semi-definite programming (SDP)\ duality and
SDP\ formulations for the max-Rains relative entropy and a channel's max-Rains
information, found recently in [Wang \textit{et al}., arXiv:1709.00200]. The
main application of our result is a single-letter, strong-converse, and
efficiently computable upper bound on the capacity of a quantum channel for
transmitting qubits when assisted by positive-partial-transpose preserving
(PPT-P) channels between every use of the channel. As the class of local
operations and classical communication (LOCC) is contained in PPT-P, our
result establishes a benchmark for the LOCC-assisted quantum capacity of an
arbitrary quantum channel, which is relevant in the context of distributed
quantum computation and quantum key distribution.

\end{abstract}

\section{Introduction}

One of the main goals of quantum information theory is to understand the
fundamental limitations on communication when a sender and receiver are
connected by a quantum communication channel \cite{H13book,H06,W15book}. Since
it might be difficult to transmit information reliably by making use of a
channel just once, a practically relevant setting is when the sender and
receiver use the channel multiple times, with the goal being to maximize the
rate of communication subject to a constraint on the error probability. The
capacity of a quantum channel is defined to be the maximum rate of reliable
communication, such that the error probability tends to zero in the limit when
the channel is utilized an arbitrary number of times.

Among the various capacities of a quantum channel $\mathcal{N}$, the
LOCC-assisted quantum capacity $Q^{\leftrightarrow}(\mathcal{N})$%
\ \cite{BDSW96}\ is particularly relevant for tasks such as distributed
quantum computation. In the setting corresponding to this capacity, the sender
and receiver are allowed to perform arbitrary LOCC\ (local operations and
classical communication) between every use of the channel, and the capacity is
equal to the maximum rate, measured in qubits per channel use, at which qubits
can be transmitted reliably from the sender to the receiver \cite{BDSW96}. Due
to the teleportation protocol \cite{PhysRevLett.70.1895}, this rate is equal
to the maximum rate at which shared entangled bits (Bell pairs) can be
generated reliably between the sender and the receiver \cite{BDSW96}. The
LOCC-assisted quantum capacity of certain channels such as the quantum erasure
channel has been known for some time \cite{PhysRevLett.78.3217}, but in
general, it remains an open question to characterize $Q^{\leftrightarrow
}(\mathcal{N})$. One can address this question by establishing either lower
bounds\ or upper bounds on $Q^{\leftrightarrow}(\mathcal{N})$.

In this paper, we are interested in placing upper bounds on the LOCC-assisted
quantum capacity, and one way of simplifying the mathematics behind this task
is to relax the class of free operations that the sender and receiver are
allowed to perform between each channel use. With this in mind, we follow the
approach of \cite{R99,R01}\ and relax the set LOCC\ to a larger class of
operations known as PPT-preserving (PPT-P), standing for channels that are
positive partial transpose preserving. The resulting capacity is then known as
the PPT-P-assisted quantum capacity $Q^{\operatorname{PPT-P},\leftrightarrow
}(\mathcal{N})$, and it is equal to the maximum rate at which qubits can be
communicated reliably from a sender to a receiver, when they are allowed to
use a PPT-preserving channel in between every use of the actual channel
$\mathcal{N}$. Figure~\ref{fig:private-code} provides a visualization of such
a PPT-P-assisted quantum communication protocol. Due to the containment LOCC
$\subset$ PPT-P \cite{R99,R01}, the inequality%
\begin{equation}
Q^{\leftrightarrow}(\mathcal{N})\leq Q^{\operatorname{PPT-P},\leftrightarrow
}(\mathcal{N}) \label{eq:LOCC-PPT-q-cap}%
\end{equation}
holds for all channels $\mathcal{N}$. Thus, if we find an upper bound on
$Q^{\operatorname{PPT-P},\leftrightarrow}(\mathcal{N})$, then by
\eqref{eq:LOCC-PPT-q-cap}, such an upper bound also bounds the physically
relevant LOCC-assisted quantum capacity $Q^{\leftrightarrow}(\mathcal{N})$.

A general approach for bounding these assisted capacities of a quantum channel
has been developed recently in \cite{KW17}\ (see
\cite{BHLS03,LHL03,Christandl2017,BGMW17,RKBKMA17} for related notions). The
starting point is to consider an entanglement measure $E(A;B)_{\rho}$
\cite{H42007}, which is evaluated for a bipartite state $\rho_{AB}$. Given
such an entanglement measure, one can define the entanglement $E(\mathcal{N})$
of a channel $\mathcal{N}$ in terms of it by taking an optimization over all
pure, bipartite states that could be input to the channel:%
\begin{equation}
E(\mathcal{N})=\sup_{\psi_{RA}}E(R;B)_{\omega}, \label{eq:channel-entang}%
\end{equation}
where $\omega_{RB}=\mathcal{N}_{A\rightarrow B}(\psi_{RA})$. The channel's
entanglement $E(\mathcal{N})$ characterizes the amount of entanglement that a
sender and receiver can generate by using the channel if they do not share
entanglement prior to its use. Due to the properties of an entanglement
measure and the well known Schmidt decomposition theorem, it suffices to take
system $R$ isomorphic to the channel input system $A$ and furthermore to
optimize over pure states $\psi_{RA}$.

\begin{figure}[ptb]
\begin{center}
\includegraphics[
width=6.5638in
]{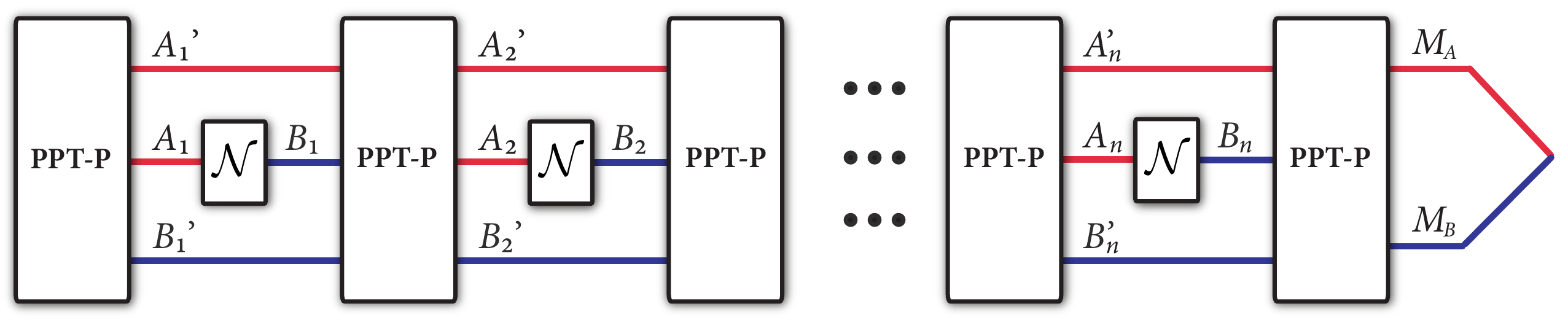}
\end{center}
\caption{A protocol for PPT-P-assisted quantum communication that uses a
quantum channel $n$ times. Every channel use is interleaved by a
PPT-preserving channel. The goal of such a protocol is to produce an
approximate maximally entangled state in the systems $M_{A}$ and $M_{B}$,
where Alice possesses system $M_{A}$ and Bob system $M_{B}$.}%
\label{fig:private-code}%
\end{figure}

One can alternatively consider the amortized entanglement $E_{A}(\mathcal{N}%
)$\ of a channel $\mathcal{N}$ as the following optimization \cite{KW17}:%
\begin{equation}
E_{A}(\mathcal{N})=\sup_{\rho_{A^{\prime}AB^{\prime}}}\left[  E(A^{\prime
};BB^{\prime})_{\tau}-E(A^{\prime}A;B^{\prime})_{\rho}\right]
,\label{eq:channel-entang-amort}%
\end{equation}
where $\tau_{A^{\prime}BB^{\prime}}=\mathcal{N}_{A\rightarrow B}%
(\rho_{A^{\prime}AB^{\prime}})$ and $\rho_{A^{\prime}AB^{\prime}}$ is a state.
The supremum is with respect to all states $\rho_{A^{\prime}AB^{\prime}}$ and
the systems $A^{\prime}B^{\prime}$ are finite-dimensional but could be
arbitrarily large (so that the supremum might never be achieved for any
particular finite-dimensional $A^{\prime}B^{\prime}$, but only in the limit of
unbounded dimension). Thus, $E_{A}(\mathcal{N})$\ is not known to be
computable in general. The amortized entanglement quantifies the net amount of
entanglement that can be generated by using the channel $\mathcal{N}$, if the
sender and receiver are allowed to begin with some initial entanglement in the
form of the state $\rho_{A^{\prime}AB^{\prime}}$. That is, $E(A^{\prime
}A;B^{\prime})_{\rho}$ quantifies the entanglement of the initial state
$\rho_{A^{\prime}AB^{\prime}}$, and $E(A^{\prime};BB^{\prime})_{\tau}$
quantifies the final entanglement of the state after the channel acts. As
observed in \cite{KW17}, the inequality%
\begin{equation}
E(\mathcal{N})\leq E_{A}(\mathcal{N})\ \label{eq:amortized-to-usual-ineq}%
\end{equation}
always holds for any entanglement measure $E$ and for any channel
$\mathcal{N}$, simply because one could take the $B^{\prime}$ system trivial
in the optimization for $E_{A}(\mathcal{N})$, which is the same as not
allowing entanglement between the sender and receiver before the channel acts.
It is nontrivial if the opposite inequality%
\begin{equation}
E_{A}(\mathcal{N})\overset{?}{\leq}E(\mathcal{N})\ \label{eq:possible-ineq}%
\end{equation}
holds, which is known to occur generally for certain entanglement
measures\ \cite{TGW14IEEE,Christandl2017,KW17}\ or for certain channels with
particular symmetries \cite{KW17}.

One of the main observations of \cite{KW17}, connected to earlier developments
in \cite{BHLS03,LHL03,Christandl2017,BGMW17,RKBKMA17}, is that the amortized
entanglement of a channel serves as an upper bound on the entanglement of the
final state $\omega_{AB}$ generated by an LOCC- or PPT-P-assisted quantum
communication protocol that uses the channel $n$ times:%
\begin{equation}
E(A;B)_{\omega}\leq nE_{A}(\mathcal{N}). \label{eq:ent-bound}%
\end{equation}
The basic intuition for why this bound holds is that, after a given channel
use, the sender and receiver are allowed to perform a free operation such as
LOCC\ or PPT, and thus the state that they share before the next channel use
could have some entanglement. So the amount of entanglement generated by each
channel use cannot exceed the amortized entanglement $E_{A}(\mathcal{N})$, and
if the channel is used $n$ times in such a protocol, then the entanglement of
the final state $\omega_{AB}$ cannot exceed the channel's amortized
entanglement multiplied by the number $n$ of channel uses. Such a general
bound can then be used to derive particular upper bounds on the assisted
quantum capacities, such as strong converse bounds. Clearly, if the inequality
in \eqref{eq:possible-ineq} holds, then $E_{A}(\mathcal{N})=E(\mathcal{N})$
and the upper bound becomes much simpler because the channel entanglement
$E(\mathcal{N})$ is simpler than the amortized entanglement $E_{A}%
(\mathcal{N})$. Thus, one of the main contributions of \cite{KW17} was to
reduce the physical question of determining meaningful upper bounds on the
assisted capacities of $\mathcal{N}$ to a purely mathematical question of
whether amortization can enhance the entanglement of a channel, i.e., whether
the equality%
\begin{equation}
E_{A}(\mathcal{N})\overset{?}{=}E(\mathcal{N})
\end{equation}
holds for a given entanglement measure $E$\ and/or channel$~\mathcal{N}$.
Furthermore, it was shown in \cite{KW17} how to incorporate the previous results of \cite{BDSW96,Mul12,P17} into the amortization framework of \cite{KW17}.

In this paper, we solve the mathematical question posed above for the
max-Rains information $R_{\max}(\mathcal{N})$ of a quantum channel
$\mathcal{N}$, by proving that amortization does not enhance it; i.e., we
prove that%
\begin{equation}
R_{\max,A}(\mathcal{N})=R_{\max}(\mathcal{N}), \label{eq:max-rains-equality}%
\end{equation}
for all channels $\mathcal{N}$, where $R_{\max,A}(\mathcal{N})$ denotes the
amortized max-Rains information. Note that $R_{\max}(\mathcal{N})$ and
$R_{\max,A}(\mathcal{N})$\ are respectively defined by taking the entanglement
measure $E$ in \eqref{eq:channel-entang}\ and
\eqref{eq:channel-entang-amort}\ to be the max-Rains relative entropy, which
we define formally in the next section. We note here that the equality in
\eqref{eq:max-rains-equality} solves an open question posed in the conclusion
of \cite{Christandl2017}, and we set our result in the context of the prior
result of \cite{Christandl2017} and other literature in
Section~\ref{sec:discussion}. The max-Rains information of a quantum channel
is a special case of a quantity known as the sandwiched R\'enyi-Rains
information \cite{TWW14}\ and was recently shown to be equal to an information
quantity discussed in \cite{WD16,WFD17} and based on semi-definite
programming. To prove our main technical result (the equality in
\eqref{eq:max-rains-equality}), we critically make use of the tools and
framework developed in the recent works \cite{WD16,WD16pra,WFD17}. In
particular, we employ semi-definite programming duality \cite{BV04} and the
well known Choi isomorphism to establish our main result, with the proof
consisting of just a few lines once the framework from
\cite{WD16,WD16pra,WFD17}\ is set in place.

The main application of the equality in \eqref{eq:max-rains-equality}\ is an
efficiently computable, single-letter, strong converse bound on
$Q^{\operatorname{PPT-P},\leftrightarrow}(\mathcal{N})$, the PPT-P-assisted
quantum capacity of an arbitrary channel $\mathcal{N}$. Due to
\eqref{eq:LOCC-PPT-q-cap}, this is also an upper bound on the physically
relevant LOCC-assisted quantum capacity $Q^{\leftrightarrow}(\mathcal{N})$. To
arrive at this result, we simply apply the general inequality in
\eqref{eq:ent-bound}\ along with the equality in
\eqref{eq:max-rains-equality}. For the benefit of the reader, we give
technical details of this application in Section~\ref{sec:strong-converse-bnd}%
. The quantity $R_{\max}(\mathcal{N})$ has already been shown in
\cite{WFD17}\ to be efficiently computable via a semi-definite program, and in
Section~\ref{sec:strong-converse-bnd}, we explain how $R_{\max}(\mathcal{N})$
is both \textquotedblleft single-letter\textquotedblright\ and
\textquotedblleft strong converse.\textquotedblright

The usefulness of the upper bound given in our paper is ultimately related with the importance of PPT-preserving channels. This is because the set of PPT-preserving channels contains the set of separable channels, and the set of separable channels strictly contains the set of LOCC channels, as shown  in \cite{PhysRevA.59.1070} and then in \cite{KANI13} for a classical scenario. Moreover, there is an entanglement monotone that can be increased by separable channels \cite{PhysRevLett.103.110502}. Thus, in general, PPT-preserving channels can increase entanglement, although this increase is not detectable by the max-Rains information. Thus, in this sense, the max-Rains information might be considered a rough measure for bounding LOCC-assisted quantum capacity. Therefore, as stressed earlier, the usefulness of our bound on the PPT-P assisted quantum capacity is directly related to PPT-preserving channels.

Our paper is organized as follows. In the next section, we review some
background material before starting with the main development.
Section~\ref{sec:technical-result} gives a short proof of our main technical
result, and Section~\ref{sec:strong-converse-bnd} discusses its application as
an efficiently computable, single-letter, strong converse bound on
$Q^{\operatorname{PPT},\leftrightarrow}(\mathcal{N})$. In
Section~\ref{sec:alt-proof-emax}, we revisit a result from
\cite{Christandl2017}, in which it was shown that amortization does not
enhance a channel's max-relative entropy of entanglement. The authors of
\cite{Christandl2017} proved this statement by employing complex interpolation
theory \cite{BL76}. We prove the main inequality underlying this statement
using a method different from that used in \cite{Christandl2017}, but along
the lines of that given for our proof of \eqref{eq:max-rains-equality} (i.e.,
convex programming duality), and we suspect that our alternative approach
could be useful in future applications. In Section~\ref{sec:discussion}, we
discuss how our result fits into the prior literature on assisted quantum
capacities and strong converses. We conclude with a brief summary in
Section~\ref{sec:conclusion}.

\section{Background and notation}

In this section, we provide background on the Choi isomorphism, partial
transpose, positive partial transpose (PPT)\ states, separable states,
PPT-preserving channels, max-relative entropy, max-Rains relative entropy, and
max-Rains information. For basic concepts and standard notation used in
quantum information theory, we point the reader to \cite{W15book}.

The Choi isomorphism represents a well known duality between channels and
states, often employed in quantum information theory. Let $\mathcal{N}%
_{A\rightarrow B}$ be a quantum channel, and let $|\Upsilon\rangle_{RA}$
denote the maximally entangled vector%
\begin{equation}
|\Upsilon\rangle_{RA}=\sum_{i}|i\rangle_{R}|i\rangle_{A},
\label{eq:max-ent-vec}%
\end{equation}
where the Hilbert spaces $\mathcal{H}_{R}$ and $\mathcal{H}_{A}$ are of the
same dimension and $\{|i\rangle_{R}\}_{i}$ and $\{|i\rangle_{A}\}_{i}$ are
fixed orthonormal bases. The Choi operator for a channel $\mathcal{N}%
_{A\rightarrow B}$ is defined as%
\begin{equation}
J_{RB}^{\mathcal{N}}=(\operatorname{id}_{R}\otimes\mathcal{N}_{A\rightarrow
B})(|\Upsilon\rangle\langle\Upsilon|_{RA}), \label{eq:Choi-op}%
\end{equation}
where $\operatorname{id}_{R}$ denotes the identity map on system $R$. One can
recover the action of the channel $\mathcal{N}_{A\rightarrow B}$ on an
arbitrary input state $\rho_{SA^{\prime}}$ as follows:%
\begin{equation}
\langle\Upsilon|_{A^{\prime}R}\ \rho_{SA^{\prime}}\otimes J_{RB}^{\mathcal{N}%
}\ |\Upsilon\rangle_{A^{\prime}R}=\mathcal{N}_{A\rightarrow B}(\rho_{SA}),
\label{eq:Choi-iso}%
\end{equation}
where $A^{\prime}$ is a system isomorphic to the channel input $A$. The above
identity can be understood in terms of a postselected variant \cite{HM04,B05}
of the quantum teleportation protocol \cite{PhysRevLett.70.1895}. Another
identity we recall is that%
\begin{equation}
\langle\Upsilon|_{RA}\left(  X_{SR}\otimes I_{A}\right)  |\Upsilon\rangle
_{RA}=\operatorname{Tr}_{R}\{X_{SR}\}, \label{eq:max-ent-trace}%
\end{equation}
for an operator $X_{SR}$ acting on $\mathcal{H}_{S}\otimes\mathcal{H}_{R}$.

For a fixed basis $\{|i\rangle_{B}\}_{i}$, the partial transpose is the
following map:%
\begin{equation}
(\operatorname{id}_{A}\otimes T_{B})(X_{AB}) =\sum_{i,j}\left(  I_{A}%
\otimes|i\rangle\langle j|_{B}\right)  X_{AB}\left(  I_{A}\otimes
|i\rangle\langle j|_{B}\right)  ,
\end{equation}
where $X_{AB}$ is an arbitrary operator acting on a tensor-product Hilbert
space $\mathcal{H}_{A}\otimes\mathcal{H}_{B}$. For simplicity we often employ
the abbreviation $T_{B}(X_{AB})=(\operatorname{id}_{A}\otimes T_{B})(X_{AB})$.
The partial transpose map plays a role in the following well known transpose
trick identity:%
\begin{equation}
\left(  X_{SR}\otimes I_{A}\right)  |\Upsilon\rangle_{RA}=\left(  T_{A}%
(X_{SA})\otimes I_{R}\right)  |\Upsilon\rangle_{RA}. \label{eq:tr-trick}%
\end{equation}
The partial transpose map plays another important role in quantum information
theory because a separable (unentangled) state%
\begin{equation}
\sigma_{AB}=\sum_{x}p(x)\tau_{A}^{x}\otimes\omega_{B}^{x}\in\operatorname{SEP}%
(A\!:\!B),
\end{equation}
for a distribution $p(x)$ and states $\tau_{A}^{x}$ and $\omega_{B}^{x}$,
stays within the set of separable states under this map
\cite{Horodecki19961,P96}:%
\begin{equation}
T_{B}(\sigma_{AB})\in\operatorname{SEP}(A\!:\!B).
\end{equation}
This motivates defining the set of PPT\ states, which are those states
$\sigma_{AB}$ for which $T_{B}(\sigma_{AB})\geq0$. This in turn motivates
defining the more general set of positive semi-definite operators
\cite{AdMVW02}:%
\begin{equation}
\operatorname{PPT}^{\prime}(A\!:\!B)=\left\{  \sigma_{AB}:\sigma_{AB}%
\geq0\wedge\left\Vert T_{B}(\sigma_{AB})\right\Vert _{1}\leq1\right\}  ,
\end{equation}
where we have employed the trace norm, defined for an operator $X$ as
$\left\Vert X\right\Vert _{1}=\operatorname{Tr}\{\left\vert X\right\vert \}$
with $\left\vert X\right\vert =\sqrt{X^{\dag}X}$. We then have the
containments $\operatorname{SEP}\subset\operatorname{PPT}\subset
\operatorname{PPT}^{\prime}$.

An LOCC\ quantum channel $\mathcal{N}_{AB\rightarrow A^{\prime}B^{\prime}}$
consists of an arbitrarily large but finite number of compositions of the following:

\begin{enumerate}
\item Alice performs a quantum instrument, which has both a quantum and
classical output. She forwards the classical output to Bob, who then performs
a quantum channel conditioned on the classical data received. This sequence of
actions corresponds to a channel of the following form:%
\begin{equation}
\sum_{x}\mathcal{F}_{A\rightarrow A^{\prime}}^{x}\otimes\mathcal{G}%
_{B\rightarrow B^{\prime}}^{x}, \label{eq-em:LOCC-channel}%
\end{equation}
where $\{\mathcal{F}_{A\rightarrow A^{\prime}}^{x}\}_{x}$ is a collection of
completely positive maps such that $\sum_{x}\mathcal{F}_{A\rightarrow
A^{\prime}}^{x}$ is a quantum channel and $\{\mathcal{G}_{B\rightarrow
B^{\prime}}^{x}\}_{x}$ is a collection of quantum channels.

\item The situation is reversed, with Bob performing the initial instrument,
who forwards the classical data to Alice, who then performs a quantum channel
conditioned on the classical data.\ This sequence of actions corresponds to a
channel of the form in \eqref{eq-em:LOCC-channel}, with the $A$ and $B$ labels switched.
\end{enumerate}

\noindent A quantum channel $\mathcal{N}_{AB\rightarrow A^{\prime}B^{\prime}}$
is a PPT-preserving channel if the map $T_{B^{\prime}}\circ\mathcal{N}%
_{AB\rightarrow A^{\prime}B^{\prime}}\circ T_{B}$ is a quantum channel
\cite{R99,R01}. Any LOCC channel is a PPT-preserving channel \cite{R99,R01}.

The max-relative entropy of a state $\rho$ relative to a positive
semi-definite operator $\sigma$ is defined as \cite{D09}%
\begin{equation}
D_{\max}(\rho\Vert\sigma)=\inf\{\lambda:\rho\leq2^{\lambda}\sigma\}.
\end{equation}
If $\operatorname{supp}(\rho)\not \subseteq \operatorname{supp}(\sigma)$, then
$D_{\max}(\rho\Vert\sigma)=\infty$. The max-relative entropy is monotone
non-increasing under the action of a quantum channel $\mathcal{N}$ \cite{D09},
in the sense that%
\begin{equation}
D_{\max}(\rho\Vert\sigma)\geq D_{\max}(\mathcal{N}(\rho)\Vert\mathcal{N}%
(\sigma)).
\end{equation}
The above inequality is also called the data-processing inequality for
max-relative entropy.

The max-Rains relative entropy of a state $\rho_{AB}$ is defined as%
\begin{equation}
R_{\max}(A;B)_{\rho}=\min_{\sigma_{AB}\in\text{PPT}^{\prime}(A:B)}D_{\max
}(\rho_{AB}\Vert\sigma_{AB}),
\end{equation}
and it is monotone non-increasing under the action of a PPT-preserving quantum
channel $\mathcal{N}_{AB\rightarrow A^{\prime}B^{\prime}}$ \cite{TWW14}, in
the sense that%
\begin{equation}
R_{\max}(A;B)_{\rho}\geq R_{\max}(A^{\prime};B^{\prime})_{\omega},
\end{equation}
for $\omega_{A^{\prime}B^{\prime}}=\mathcal{N}_{AB\rightarrow A^{\prime
}B^{\prime}}(\rho_{AB})$. The max-Rains information of a quantum channel
$\mathcal{N}_{A\rightarrow B}$\ is defined by replacing $E$ in
\eqref{eq:channel-entang}\ with the max-Rains relative entropy $R_{\max}$;
i.e.,
\begin{equation}
R_{\max}(\mathcal{N})=\max_{\phi_{SA}}R_{\max}(S;B)_{\omega},
\end{equation}
where $\omega_{SB}=\mathcal{N}_{A\rightarrow B}(\phi_{SA})$ and $\phi_{SA}$ is
a pure state, with $\left\vert S\right\vert =\left\vert A\right\vert $. The
amortized max-Rains information of a channel, denoted as $R_{\max
,A}(\mathcal{N})$, is defined by replacing $E$ in
\eqref{eq:channel-entang-amort} with the max-Rains relative entropy $R_{\max}$.

Recently, in \cite[Eq.~(8)]{WD16pra} (see also \cite[Eq.~(36)]{WFD17}), the
max-Rains relative entropy of a state $\rho_{AB}$ was expressed as%
\begin{equation}
R_{\max}(A;B)_{\rho}=\log_{2}W(A;B)_{\rho},\label{eq:alt-max-Rains-state}%
\end{equation}
where $W(A;B)_{\rho}$ is the solution to the following semi-definite program:%
\begin{align}
\text{minimize } &  \ \operatorname{Tr}\{C_{AB}+D_{AB}\}\nonumber\\
\text{subject to } &  \ C_{AB},\ D_{AB}\geq0,\nonumber\\
&  \ T_{B}(C_{AB}-D_{AB})\geq\rho_{AB}.\label{eq:sdp-W}%
\end{align}
Similarly, in \cite[Eq.~(21)]{WFD17}, the max-Rains information of a quantum
channel $\mathcal{N}_{A\rightarrow B}$ was expressed as%
\begin{equation}
R_{\max}(\mathcal{N})=\log\Gamma(\mathcal{N}),\label{eq:alt-max-Rains-channel}%
\end{equation}
where $\Gamma(\mathcal{N})$ is the solution to the following semi-definite
program:%
\begin{align}
\text{minimize } &  \ \left\Vert \operatorname{Tr}_{B}\{V_{SB}+Y_{SB}%
\}\right\Vert _{\infty}\nonumber\\
\text{subject to\ } &  \ Y_{SB},\ V_{SB}\geq0,\nonumber\\
&  \ T_{B}(V_{SB}-Y_{SB})\geq J_{SB}^{\mathcal{N}}.\label{eq:Gamma-sdp}%
\end{align}
These formulations of $R_{\max}(A;B)_{\rho}$ and $R_{\max}(\mathcal{N})$\ are
the tools that we use to prove our main technical result,
Proposition~\ref{prop:amort-doesnt-help}. It is worthwhile to mention that the
formulations above follow by employing the theory of semi-definite programming
and its duality.

\section{Main technical result}

\label{sec:technical-result}The following proposition constitutes our main
technical result, and an immediate corollary of it is that amortization does
not enhance the max-Rains information of a quantum channel:

\begin{proposition}
\label{prop:amort-doesnt-help}Let $\rho_{A^{\prime}AB^{\prime}}$ be a state
and let $\mathcal{N}_{A\rightarrow B}$ be a quantum channel. Then%
\begin{equation}
R_{\max}(A^{\prime};BB^{\prime})_{\omega}\leq R_{\max}(\mathcal{N})+R_{\max
}(A^{\prime}A;B^{\prime})_{\rho}, \label{eq:r-max-amort-no-help-1}%
\end{equation}
where%
\begin{equation}
\omega_{A^{\prime}BB^{\prime}}=\mathcal{N}_{A\rightarrow B}(\rho_{A^{\prime
}AB^{\prime}}).
\end{equation}

\end{proposition}

\begin{proof}
By removing logarithms and applying \eqref{eq:alt-max-Rains-state} and
\eqref{eq:alt-max-Rains-channel}, the desired inequality is equivalent to the
following one:%
\begin{equation}
W(A^{\prime};BB^{\prime})_{\omega}\leq\Gamma(\mathcal{N})\cdot W(A^{\prime
}A;B^{\prime})_{\rho}, \label{eq:amortized-ineq}%
\end{equation}
and so we aim to prove this one. Exploiting the identity in \eqref{eq:sdp-W},
we find that%
\begin{equation}
W(A^{\prime}A;B^{\prime})_{\rho}=\min\operatorname{Tr}\{C_{A^{\prime
}AB^{\prime}}+D_{A^{\prime}AB^{\prime}}\},
\end{equation}
subject to the constraints%
\begin{align}
C_{A^{\prime}AB^{\prime}},\ D_{A^{\prime}AB^{\prime}}  &  \geq0,\\
T_{B^{\prime}}(C_{A^{\prime}AB^{\prime}}-D_{A^{\prime}AB^{\prime}})  &
\geq\rho_{A^{\prime}AB^{\prime}}, \label{eq:W-constraint}%
\end{align}
while the identity in \eqref{eq:Gamma-sdp} gives that%
\begin{equation}
\Gamma(\mathcal{N})=\min\left\Vert \operatorname{Tr}_{B}\{V_{SB}%
+Y_{SB}\}\right\Vert _{\infty},
\end{equation}
subject to the constraints%
\begin{align}
Y_{SB},V_{SB}  &  \geq0,\\
T_{B}(V_{SB}-Y_{SB})  &  \geq J_{SB}^{\mathcal{N}}.
\label{eq:Gamma-constraint}%
\end{align}
The identity in \eqref{eq:sdp-W} implies that the left-hand side of
\eqref{eq:amortized-ineq}\ is equal to%
\begin{equation}
W(A^{\prime};BB^{\prime})_{\omega}=\min\operatorname{Tr}\{E_{A^{\prime
}BB^{\prime}}+F_{A^{\prime}BB^{\prime}}\},
\end{equation}
subject to the constraints%
\begin{align}
E_{A^{\prime}BB^{\prime}},F_{A^{\prime}BB^{\prime}}  &  \geq
0,\label{eq:W-input-const-1}\\
\mathcal{N}_{A\rightarrow B}(\rho_{A^{\prime}AB^{\prime}})  &  \leq
T_{BB^{\prime}}(E_{A^{\prime}BB^{\prime}}-F_{A^{\prime}BB^{\prime}}).
\label{eq:W-input-const-2}%
\end{align}

With these SDP formulations in place, we can now establish the inequality in
\eqref{eq:amortized-ineq} by making judicious choices for $E_{A^{\prime
}BB^{\prime}}$ and $F_{A^{\prime}BB^{\prime}}$. Let $C_{A^{\prime}AB^{\prime}%
}$ and $D_{A^{\prime}AB^{\prime}}$ be optimal for $W(A^{\prime}A;B^{\prime
})_{\rho}$, and let $Y_{SB}$ and $V_{SB}$ be optimal for $\Gamma(\mathcal{N}%
)$. Let $|\Upsilon\rangle_{SA}$ be the maximally entangled vector, as defined
in \eqref{eq:max-ent-vec}. Pick%
\begin{align*}
E_{A^{\prime}BB^{\prime}}  &  =\langle\Upsilon|_{SA}C_{A^{\prime}AB^{\prime}%
}\otimes V_{SB}+D_{A^{\prime}AB^{\prime}}\otimes Y_{SB}|\Upsilon\rangle
_{SA},\\
F_{A^{\prime}BB^{\prime}}  &  =\langle\Upsilon|_{SA}C_{A^{\prime}AB^{\prime}%
}\otimes Y_{SB}+D_{A^{\prime}AB^{\prime}}\otimes V_{SB}|\Upsilon\rangle_{SA}.
\end{align*}
We note that these choices are somewhat similar to those made in the proof of
\cite[Proposition~6]{WFD17}, and they can be understood roughly via
\eqref{eq:Choi-iso} as a postselected teleportation of the optimal operators
of $W(A^{\prime}A;B^{\prime})_{\rho}$ through the optimal operators of
$\Gamma(\mathcal{N})$, with the optimal operators of $W(A^{\prime}A;B^{\prime
})_{\rho}$ being in correspondence with the input state $\rho_{A^{\prime
}AB^{\prime}}$ through \eqref{eq:W-constraint} and the optimal operators of
$\Gamma(\mathcal{N})$ being in correspondence with the Choi operator
$J_{SB}^{\mathcal{N}}$ through \eqref{eq:Gamma-constraint}. We then have that
$E_{A^{\prime}BB^{\prime}},F_{A^{\prime}BB^{\prime}}\geq0$ because
$C_{A^{\prime}AB^{\prime}}$, $D_{A^{\prime}AB^{\prime}}$, $Y_{SB}$,
$V_{SB}\geq0$. Consider that%
\begin{align}
T_{BB^{\prime}}(E_{A^{\prime}BB^{\prime}}-F_{A^{\prime}BB^{\prime}})  &
=T_{BB^{\prime}}\left[  \langle\Upsilon|_{SA}(C_{A^{\prime}AB^{\prime}%
}-D_{A^{\prime}AB^{\prime}})\otimes(V_{SB}-Y_{SB})|\Upsilon\rangle_{SA}\right]
\nonumber\\
&  =\langle\Upsilon|_{SA}T_{B^{\prime}}(C_{A^{\prime}AB^{\prime}}%
-D_{A^{\prime}AB^{\prime}})\otimes T_{B}(V_{SB}-Y_{SB})|\Upsilon\rangle
_{SA}\nonumber\\
&  \geq\langle\Upsilon|_{SA}\ \rho_{A^{\prime}AB^{\prime}}\otimes
J_{SB}^{\mathcal{N}}|\Upsilon\rangle_{SA}\nonumber\\
&  =\mathcal{N}_{A\rightarrow B}(\rho_{A^{\prime}AB^{\prime}}).
\end{align}
The inequality follows from \eqref{eq:W-constraint} and
\eqref{eq:Gamma-constraint}, and the last equality follows from
\eqref{eq:Choi-iso}. Also consider that%
\begin{align}
\operatorname{Tr}\{E_{A^{\prime}BB^{\prime}}+F_{A^{\prime}BB^{\prime}}\}  &
=\operatorname{Tr}\{\langle\Upsilon|_{SA}(C_{A^{\prime}AB^{\prime}%
}+D_{A^{\prime}AB^{\prime}})\otimes(V_{SB}+Y_{SB})|\Upsilon\rangle
_{SA}\}\nonumber\\
&  =\operatorname{Tr}\{(C_{A^{\prime}AB^{\prime}}+D_{A^{\prime}AB^{\prime}%
})T_{A}(V_{AB}+Y_{AB})\}\nonumber\\
&  =\operatorname{Tr}\{(C_{A^{\prime}AB^{\prime}}+D_{A^{\prime}AB^{\prime}%
})T_{A}(\operatorname{Tr}_{B}\left\{  V_{AB}+Y_{AB}\right\}  )\}\nonumber\\
&  \leq\operatorname{Tr}\{C_{A^{\prime}AB^{\prime}}+D_{A^{\prime}AB^{\prime}%
}\}\left\Vert T_{A}(\operatorname{Tr}_{B}\left\{  V_{AB}+Y_{AB}\right\}
)\right\Vert _{\infty}\nonumber\\
&  =\operatorname{Tr}\{C_{A^{\prime}AB^{\prime}}+D_{A^{\prime}AB^{\prime}%
}\}\left\Vert \operatorname{Tr}_{B}\left\{  V_{AB}+Y_{AB}\right\}  \right\Vert
_{\infty}\nonumber\\
&  =W(A^{\prime}A;B^{\prime})_{\rho}\cdot\Gamma(\mathcal{N}).
\label{eq:final-r-max-reasoning-amort-no-help}%
\end{align}
The second equality follows from \eqref{eq:tr-trick} and
\eqref{eq:max-ent-trace}. The inequality is a consequence of H\"{o}lder's
inequality. The final equality follows because the spectrum of an operator is
invariant under the action of a (full)\ transpose (note, in this case, that
$T_{A}$ is a full transpose because the operator $\operatorname{Tr}%
_{B}\left\{  V_{AB}+Y_{AB}\right\}  $ acts only on system $A$).

Thus, we can conclude that our choices of $E_{A^{\prime}BB^{\prime}}$ and
$F_{A^{\prime}BB^{\prime}}$ are feasible for $W(A^{\prime};BB^{\prime
})_{\omega}$. Since $W(A^{\prime};BB^{\prime})_{\omega}$ involves a
minimization over all $E_{A^{\prime}BB^{\prime}}$ and $F_{A^{\prime}%
BB^{\prime}}$ satisfying \eqref{eq:W-input-const-1} and
\eqref{eq:W-input-const-2}, this concludes our proof of \eqref{eq:amortized-ineq}.
\end{proof}

An immediate corollary of Proposition~\ref{prop:amort-doesnt-help}\ is the following:

\begin{corollary}
\label{cor:amort-not-enhance}Amortization does not enhance the max-Rains
information of a quantum channel $\mathcal{N}_{A\rightarrow B}$; i.e., the
following equality holds%
\begin{equation}
R_{\max,A}(\mathcal{N})=R_{\max}(\mathcal{N}).
\end{equation}

\end{corollary}

\begin{proof}
The inequality $R_{\max,A}(\mathcal{N})\geq R_{\max}(\mathcal{N})$ always
holds, as reviewed in \eqref{eq:amortized-to-usual-ineq}. The other inequality
is an immediate consequence of Proposition~\ref{prop:amort-doesnt-help}.
Letting $\rho_{A^{\prime}AB^{\prime}}$ denote an arbitrary input state,
Proposition~\ref{prop:amort-doesnt-help} implies that%
\begin{equation}
R_{\max}(A^{\prime};BB^{\prime})_{\omega}-R_{\max}(A^{\prime}A;B^{\prime
})_{\rho}\leq R_{\max}(\mathcal{N}),
\end{equation}
where $\omega_{A^{\prime}BB^{\prime}}=\mathcal{N}_{A\rightarrow B}%
(\rho_{A^{\prime}AB^{\prime}})$. Since the inequality holds for any state
$\rho_{A^{\prime}AB^{\prime}}$, it holds for the supremum over all such input
states, leading to $R_{\max,A}(\mathcal{N})\leq R_{\max}(\mathcal{N})$.
\end{proof}

\section{Application to PPT-P-assisted quantum communication}

\label{sec:strong-converse-bnd}We now give our main application of
Proposition~\ref{prop:amort-doesnt-help}, which is that the max-Rains
information is a single-letter, strong-converse upper bound on the
PPT-P-assisted quantum capacity of any channel. The term \textquotedblleft
single-letter\textquotedblright\ refers to the fact that the max-Rains
information requires an optimization over a single use of the channel. As we
remarked previously, the max-Rains information is efficiently computable via
semi-definite programming, as observed in \cite{WD16,WFD17}. Finally, the
bound is a strong converse bound because, as we will show, if the rate of a
sequence of PPT-P-assisted quantum communication protocols exceeds the
max-Rains information, then the error probability of these protocols
necessarily tends to one exponentially fast in the number of channel uses.

\subsection{Protocol for PPT-P-assisted quantum communication}

\label{sec:PPT-protocol}We begin by reviewing the structure of a
PPT-P-assisted quantum communication protocol, along the lines discussed in
\cite{KW17}. In such a protocol, a sender Alice and a receiver Bob are
spatially separated and connected by a quantum channel $\mathcal{N}%
_{A\rightarrow B}$. They begin by performing a PPT-P channel $\mathcal{P}%
_{\emptyset\rightarrow A_{1}^{\prime}A_{1}B_{1}^{\prime}}^{(1)}$, which leads
to a PPT state $\rho_{A_{1}^{\prime}A_{1}B_{1}^{\prime}}^{(1)}$, where
$A_{1}^{\prime}$ and $B_{1}^{\prime}$ are systems that are finite-dimensional
but arbitrarily large. The system $A_{1}$ is such that it can be fed into the
first channel use. Alice sends system $A_{1}$ through the first channel use,
leading to a state $\sigma_{A_{1}^{\prime}B_{1}B_{1}^{\prime}}^{(1)}%
\equiv\mathcal{N}_{A_{1}\rightarrow B_{1}}(\rho_{A_{1}^{\prime}A_{1}%
B_{1}^{\prime}}^{(1)})$. Alice and Bob then perform the PPT-P channel
$\mathcal{P}_{A_{1}^{\prime}B_{1}B_{1}^{\prime}\rightarrow A_{2}^{\prime}%
A_{2}B_{2}^{\prime}}^{(2)}$, which leads to the state%
\begin{equation}
\rho_{A_{2}^{\prime}A_{2}B_{2}^{\prime}}^{(2)}\equiv\mathcal{P}_{A_{1}%
^{\prime}B_{1}B_{1}^{\prime}\rightarrow A_{2}^{\prime}A_{2}B_{2}^{\prime}%
}^{(2)}(\sigma_{A_{1}^{\prime}B_{1}B_{1}^{\prime}}^{(1)}).
\end{equation}
Alice sends system $A_{2}$ through the second channel use $\mathcal{N}%
_{A_{2}\rightarrow B_{2}}$, leading to the state $\sigma_{A_{2}^{\prime}%
B_{2}B_{2}^{\prime}}^{(2)}\equiv\mathcal{N}_{A_{2}\rightarrow B_{2}}%
(\rho_{A_{2}^{\prime}A_{2}B_{2}^{\prime}}^{(1)})$. This process iterates:\ the
protocol uses the channel $n$ times. In general, we have the following states
for all $i\in\{2,\ldots,n\}$:%
\begin{align}
\rho_{A_{i}^{\prime}A_{i}B_{i}^{\prime}}^{(i)}  &  \equiv\mathcal{P}%
_{A_{i-1}^{\prime}B_{i-1}B_{i-1}^{\prime}\rightarrow A_{i}^{\prime}A_{i}%
B_{i}^{\prime}}^{(i)}(\sigma_{A_{i-1}^{\prime}B_{i-1}B_{i-1}^{\prime}}%
^{(i-1)}),\\
\sigma_{A_{i}^{\prime}B_{i}B_{i}^{\prime}}^{(i)}  &  \equiv\mathcal{N}%
_{A_{i}\rightarrow B_{i}}(\rho_{A_{i}^{\prime}A_{i}B_{i}^{\prime}}^{(i)}),
\end{align}
where $\mathcal{P}_{A_{i-1}^{\prime}B_{i-1}B_{i-1}^{\prime}\rightarrow
A_{i}^{\prime}A_{i}B_{i}^{\prime}}^{(i)}$ is a PPT\ channel. The final step of
the protocol consists of a PPT-P channel $\mathcal{P}_{A_{n}^{\prime}%
B_{n}B_{n}^{\prime}\rightarrow M_{A}M_{B}}^{(n+1)}$, which generates the
systems $M_{A}$ and $M_{B}$ for Alice and Bob, respectively. The protocol's
final state is as follows:%
\begin{equation}
\omega_{M_{A}M_{B}}\equiv\mathcal{P}_{A_{n}^{\prime}B_{n}B_{n}^{\prime
}\rightarrow M_{A}M_{B}}^{(n+1)}(\sigma_{A_{n}^{\prime}B_{n}B_{n}^{\prime}%
}^{(n)}).
\end{equation}
Figure~\ref{fig:private-code}\ depicts such a protocol.

The goal of the protocol is that the final state $\omega_{M_{A}M_{B}}$ is
close to a maximally entangled state. Fix $n,M\in\mathbb{N}$ and
$\varepsilon\in\lbrack0,1]$. The original protocol is an $(n,M,\varepsilon)$
protocol if the channel is used $n$ times as discussed above, $\left\vert
M_{A}\right\vert =\left\vert M_{B}\right\vert =M$, and if%
\begin{align}
F(\omega_{M_{A}M_{B}},\Phi_{M_{A}M_{B}})  &  =\langle\Phi|_{M_{A}M_{B}}%
\omega_{M_{A}M_{B}}|\Phi\rangle_{M_{A}M_{B}}\\
&  \geq1-\varepsilon, \label{eq:fidelity-assump}%
\end{align}
where the fidelity $F(\tau,\kappa)\equiv\left\Vert \sqrt{\tau}\sqrt{\kappa
}\right\Vert _{1}^{2}$ \cite{U76} and the maximally entangled state
$\Phi_{M_{A}M_{B}}=|\Phi\rangle\langle\Phi|_{M_{A}M_{B}}$ is defined from%
\begin{equation}
|\Phi\rangle_{M_{A}M_{B}}\equiv\frac{1}{\sqrt{M}}\sum_{m=1}^{M}|m\rangle
_{M_{A}}\otimes|m\rangle_{M_{B}}.
\end{equation}

A rate $R$ is achievable for PPT-P-assisted quantum communication if for all
$\varepsilon\in(0,1]$, $\delta>0$, and sufficiently large$~n$, there exists an
$(n,2^{n\left(  R-\delta\right)  },\varepsilon)$ protocol. The PPT-P-assisted
quantum capacity of a channel $\mathcal{N}$, denoted as
$Q^{\operatorname{PPT-P},\leftrightarrow}(\mathcal{N})$, is equal to the
supremum of all achievable rates.

On the other hand, a rate $R$ is a strong converse rate for PPT-P-assisted
quantum communication if for all $\varepsilon\in[0,1)$, $\delta>0$, and
sufficiently large$~n$, there does not exist an $(n,2^{n\left(  R+\delta
\right)  },\varepsilon)$ protocol. The strong converse PPT-P-assisted quantum
capacity $Q^{\operatorname{PPT-P},\leftrightarrow\dag}(\mathcal{N})$ is equal
to the infimum of all strong converse rates. We say that a channel obeys the
strong converse property for PPT-P-assisted quantum communication if
$Q^{\operatorname{PPT-P},\leftrightarrow}(\mathcal{N})=Q^{\operatorname{PPT-P}%
,\leftrightarrow\dag}(\mathcal{N})$.

We can also consider the whole development above when we only allow the
assistance of LOCC\ channels instead of PPT\ channels. In this case, we have
similar notions as above, and then we arrive at the LOCC-assisted quantum
capacity $Q^{\leftrightarrow}(\mathcal{N})$ and the strong converse
LOCC-assisted quantum capacity $Q^{\leftrightarrow\dag}(\mathcal{N})$. It then
immediately follows that%
\begin{align}
Q^{\leftrightarrow}(\mathcal{N})  &  \leq Q^{\operatorname{PPT-P}%
,\leftrightarrow}(\mathcal{N}) ,\\
Q^{\leftrightarrow\dag}(\mathcal{N})  &  \leq Q^{\operatorname{PPT-P}%
,\leftrightarrow\dag}(\mathcal{N})
\end{align}
because every LOCC\ channel is a PPT\ channel.

\subsection{Max-Rains information as a strong converse rate for PPT-P-assisted
quantum communication}

We now prove the following upper bound on the communication rate $\frac{1}%
{n}\log_{2}M$ (qubits per channel use) of any $(n,M,\varepsilon)$
PPT-P-assisted protocol:

\begin{theorem}
Fix $n,M\in\mathbb{N}$ and $\varepsilon\in(0,1)$. The following bound holds
for an $(n,M,\varepsilon)$ protocol for PPT-P-assisted quantum communication
over a quantum channel $\mathcal{N}$:%
\begin{equation}
\log_{2}M\leq n R_{\max}(\mathcal{N})+\log_{2}\!\left(  \frac{1}%
{1-\varepsilon}\right)  . \label{eq:max-Rains-converse-bound}%
\end{equation}

\end{theorem}

\begin{proof}
For convenience of the reader, we give a complete proof, but we note that some
of the essential steps are available in prior works
\cite{Christandl2017,RKBKMA17,KW17}. From the assumption in
\eqref{eq:fidelity-assump}, it follows that%
\begin{equation}
\operatorname{Tr}\{\Phi_{M_{A}M_{B}}\omega_{M_{A}M_{B}}\}\geq1-\varepsilon,
\end{equation}
while \cite[Lemma~2]{R99} implies that%
\begin{equation}
\operatorname{Tr}\{\Phi_{M_{A}M_{B}}\sigma_{M_{A}M_{B}}\}\leq\frac{1}%
{M},\label{eq:ent-test-PPT-states}%
\end{equation}
for all $\sigma_{M_{A}M_{B}}\in\operatorname{PPT}^{\prime}(M_{A}\!:\!M_{B})$.
So under an \textquotedblleft entanglement test,\textquotedblright\ i.e., a
measurement of the form $\left\{  \Phi_{M_{A}M_{B}},I_{M_{A}M_{B}}-\Phi
_{M_{A}M_{B}}\right\}  $ and applying the data processing inequality for the
max-relative entropy, we find for all $\sigma_{M_{A}M_{B}}\in
\operatorname{PPT}^{\prime}(M_{A}\!:\!M_{B})$ that%
\begin{align}
D_{\max}(\omega_{M_{A}M_{B}}\Vert\sigma_{M_{A}M_{B}})  & \geq D_{\max
}(\{p,1-p\}\Vert\{q,\operatorname{Tr}\{\sigma_{M_{A}M_{B}}\}-q\})\\
& =\log_{2}\max\{p/q,(1-p)/(\operatorname{Tr}\{\sigma_{M_{A}M_{B}}\}-q)\}\\
& \geq\log_{2}(p/q)\\
& \geq\log_{2}\left[  \left(  1-\varepsilon\right)  M\right]  ,
\end{align}
where $p\equiv\operatorname{Tr}\{\Phi_{M_{A}M_{B}}\omega_{M_{A}M_{B}}\}$ and
$q=\operatorname{Tr}\{\Phi_{M_{A}M_{B}}\sigma_{M_{A}M_{B}}\}$. Since the above
chain of inequalities holds for all $\sigma_{M_{A}M_{B}}\in\operatorname{PPT}%
^{\prime}(M_{A}\!:\!M_{B})$, we conclude that%
\begin{equation}
R_{\max}(M_{A};M_{B})_{\omega}\geq\log_{2}\left[  \left(  1-\varepsilon
\right)  M\right]  .\label{eq:first-bound-final-state}%
\end{equation}
From the monotonicity of the Rains relative entropy with respect to
PPT-preserving channels \cite{R01,TWW14}, we find that%
\begin{align}
R_{\max}(M_{A};M_{B})_{\omega} &  \leq R_{\max}(A_{n}^{\prime};B_{n}%
B_{n}^{\prime})_{\sigma^{(n)}}\\
&  =R_{\max}(A_{n}^{\prime};B_{n}B_{n}^{\prime})_{\sigma^{(n)}}-R_{\max}%
(A_{1}^{\prime}A_{1};B_{1}^{\prime})_{\rho^{(1)}}\\
&  =R_{\max}(A_{n}^{\prime};B_{n}B_{n}^{\prime})_{\sigma^{(n)}}+\left[
\sum_{i=2}^{n}R_{\max}(A_{i}^{\prime}A_{i};B_{i}^{\prime})_{\rho^{(i)}%
}-R_{\max}(A_{i}^{\prime}A_{i};B_{i}^{\prime})_{\rho^{(i)}}\right]
\nonumber\\
&  \qquad-R_{\max}(A_{1}^{\prime}A_{1};B_{1}^{\prime})_{\rho^{(1)}}\\
&  \leq\sum_{i=1}^{n}\left[  R_{\max}(A_{i}^{\prime};B_{i}B_{i}^{\prime
})_{\sigma^{(i)}}-R_{\max}(A_{i}^{\prime}A_{i};B_{i}^{\prime})_{\rho^{(i)}%
}\right]  \\
&  \leq nR_{\max}(\mathcal{N}).\label{eq:last-bound-amortized-proof}%
\end{align}
The first equality follows because the state $\rho_{A_{1}^{\prime}A_{1}%
B_{1}^{\prime}}^{(1)}$ is a PPT state with vanishing max-Rains relative
entropy. The second equality follows by adding and subtracting terms. The
second inequality follows because $R_{\max}(A_{i}^{\prime}A_{i};B_{i}^{\prime
})_{\rho^{(i)}}\leq R_{\max}(A_{i-1}^{\prime};B_{i-1}B_{i-1}^{\prime}%
)_{\sigma^{(i-1)}}$ for all $i\in\{2,\ldots,n\}$, due to monotonicity of the
Rains relative entropy with respect to PPT-P channels. The final inequality
follows by applying Proposition~\ref{prop:amort-doesnt-help}\ to each term
$R_{\max}(A_{n}^{\prime};B_{n}B_{n}^{\prime})_{\sigma^{(i)}}-R_{\max}%
(A_{i}^{\prime}A_{i};B_{i}^{\prime})_{\rho^{(i)}}$. Combining
\eqref{eq:first-bound-final-state}\ and \eqref{eq:last-bound-amortized-proof},
we arrive at the inequality in \eqref{eq:max-Rains-converse-bound}.
\end{proof}

\begin{remark}
The bound in \eqref{eq:max-Rains-converse-bound}\ can also be rewritten in the
following way:%
\begin{equation}
1-\varepsilon\leq2^{-n\left[  Q-R_{\max}(\mathcal{N})\right]  },
\end{equation}
where we set the rate $Q=\frac{1}{n}\log_{2}M$. Thus, if the communication
rate $Q$ is strictly larger than the max-Rains information $R_{\max
}(\mathcal{N})$, then the fidelity of the transmission ($1-\varepsilon$)
decays exponentially fast to zero in the number $n$ of channel uses.
\end{remark}

An immediate corollary of the above is the following strong converse statement:

\begin{corollary}
The strong converse PPT-P-assisted quantum capacity is bounded from above by
the max-Rains information:%
\begin{equation}
Q^{\operatorname{PPT-P},\leftrightarrow\dag}(\mathcal{N})\leq R_{\max
}(\mathcal{N}).
\end{equation}

\end{corollary}

\section{Amortization does not increase a channel's max-relative entropy of
entanglement}

\label{sec:alt-proof-emax}One of the main results of \cite{Christandl2017} is
that amortization does not increase a channel's max-relative entropy of
entanglement; i.e.,%
\begin{equation}
E_{\max,A}(\mathcal{N})=E_{\max}(\mathcal{N}),
\label{eq:max-rel-amort-no-help}%
\end{equation}
where $E_{\max}(\mathcal{N})$ denotes a channel's max-relative entropy of
entanglement (we will define this shortly). The authors of
\cite{Christandl2017} proved \eqref{eq:max-rel-amort-no-help}\ by employing
the methods of complex interpolation \cite{BL76}. The main application of
\eqref{eq:max-rel-amort-no-help} is that $E_{\max}(\mathcal{N})$ is a strong
converse upper bound on the secret-key-agreement capacity of a quantum channel
\cite{Christandl2017}\ (this is defined as the private capacity of the
channel, when arbitrary LOCC is allowed between every channel use---see
\cite{WTB16} or \cite{Christandl2017}\ for a definition).

In this section, we provide an alternate proof of
\eqref{eq:max-rel-amort-no-help}, which is along the lines of the proofs of
Proposition~\ref{prop:amort-doesnt-help}\ and
Corollary~\ref{cor:amort-not-enhance}. We think that this approach brings a
different perspective to the result of \cite{Christandl2017} and could
potentially be useful in future applications.

To begin with, let us recall the definition of the max-relative entropy of
entanglement of a bipartite state $\rho_{AB}$ \cite{D09}:%
\begin{equation}
E_{\max}(A;B)_{\rho}=\min_{\sigma_{AB}\in\operatorname{SEP}(A:B)}D_{\max}%
(\rho_{AB}\Vert\sigma_{AB}). \label{eq:e-max}%
\end{equation}
Let $\overrightarrow{\operatorname{SEP}}(A:B)$ denote the cone of all
separable operators, i.e., $X_{AB}\in\overrightarrow{\operatorname{SEP}%
}(A\!:\!B)$ if there exists a positive integer $L$ and positive semi-definite
operators $\{P_{A}^{x}\}_{x}$ and $\{Q_{B}^{x}\}_{x}$ such that $X_{AB}%
=\sum_{x=1}^{L}P_{A}^{x}\otimes Q_{B}^{x}$. The arrow in $\overrightarrow
{\operatorname{SEP}}(A:B)$ is meant to remind the reader of ``cone'' and is
not intended to indicate any directionality between the $A$ and $B$ systems.
In what follows, we sometimes employ the shorthands $\operatorname{SEP}$ and
$\overrightarrow{\operatorname{SEP}}$ when the bipartite cuts are clear from
the context. Then we have the following alternative expression for the
max-relative entropy of entanglement:

\begin{lemma}
\label{lem:alt-emax}Let $\rho_{AB}$ be a bipartite state. Then%
\begin{equation}
E_{\max}(A;B)_{\rho}=\log_{2}W_{\operatorname{sep}}(A;B)_{\rho},
\end{equation}
where%
\begin{equation}
W_{\operatorname{sep}}(A;B)_{\rho}=\min_{X_{AB}\in\overrightarrow
{\operatorname{SEP}}}\left\{  \operatorname{Tr}\{X_{AB}\}:\rho_{AB}\leq
X_{AB}\right\}  .
\end{equation}

\end{lemma}

\begin{proof}
Employing the definition in \eqref{eq:e-max}, consider that%
\begin{align}
\min_{\sigma_{AB}\in\operatorname{SEP}(A:B)}D_{\max}(\rho_{AB}\Vert\sigma
_{AB})  &  =\log_{2}\min_{\mu,\sigma_{AB}}\{\mu:\rho_{AB}\leq\mu\sigma
_{AB},\ \sigma_{AB}\in\operatorname{SEP}\}\\
&  =\log_{2}\min_{X_{AB}}\left\{  \operatorname{Tr}\{X_{AB}\}:\rho_{AB}\leq
X_{AB},\ X_{AB}\in\overrightarrow{\operatorname{SEP}}\right\}  .
\end{align}
This concludes the proof.
\end{proof}

\bigskip

We can then define a channel's max-relative entropy of entanglement $E_{\max
}(\mathcal{N})$ as in \eqref{eq:channel-entang},\ by replacing $E$ with
$E_{\max}$. We can alternatively write $E_{\max}(\mathcal{N})$ as follows, by
employing similar reasoning as given in the proof of \cite[Lemma~6]{CMW14}:%
\begin{equation}
E_{\max}(\mathcal{N})=\max_{\rho_{S}}\min_{\sigma_{SB}\in\operatorname{SEP}%
}D_{\max}(\rho_{S}^{1/2}J_{SB}^{\mathcal{N}}\rho_{S}^{1/2}\Vert\sigma_{SB}),
\label{eq:e_max-channel-alt}%
\end{equation}
where $\rho_{S}$ is a density operator and $J_{SB}^{\mathcal{N}}$ is the Choi
operator for the channel $\mathcal{N}$, as defined in \eqref{eq:Choi-op}. We
now prove the following alternative expression for $E_{\max}(\mathcal{N})$:

\begin{lemma}
\label{lem:alt-e-max-channel}Let $\mathcal{N}_{A\rightarrow B}$ be a quantum
channel. Then%
\begin{equation}
E_{\max}(\mathcal{N})=\log_{2}\Sigma(\mathcal{N}),
\end{equation}
where%
\begin{equation}
\Sigma(\mathcal{N})=\min_{Y_{SB}\in\overrightarrow{\operatorname{SEP}}%
}\left\{  \left\Vert \operatorname{Tr}_{B}\{Y_{SB}\}\right\Vert _{\infty
}:J_{SB}^{\mathcal{N}}\leq Y_{SB}\right\}  .
\end{equation}

\end{lemma}

\begin{proof}
Employing \eqref{eq:e_max-channel-alt} and Lemma~\ref{lem:alt-emax}, we find
that%
\begin{equation}
E_{\max}(\mathcal{N})= \log\max_{\rho_{S}}\min_{Y_{SB}\in\overrightarrow
{\operatorname{SEP}}}\left\{  \operatorname{Tr}\{Y_{SB}\}:\rho_{S}^{1/2}%
J_{SB}^{\mathcal{N}}\rho_{S}^{1/2}\leq Y_{SB}\right\}  .
\end{equation}
So our aim is to prove that the expression inside the logarithm is equal to
$\Sigma(\mathcal{N})$. Taking the ansatz that $\rho_{S}$ is an invertible
density operator, we find that the condition $\rho_{S}^{1/2}J_{SB}%
^{\mathcal{N}}\rho_{S}^{1/2}\leq Y_{SB}$ is equivalent to the condition
$J_{SB}^{\mathcal{N}}\leq\rho_{S}^{-1/2}Y_{SB}\rho_{S}^{-1/2}=Y_{SB}^{\prime
}\in\overrightarrow{\operatorname{SEP}}(S\!:\!B)$. Noting that $Y_{SB}%
=\rho_{S}^{1/2}Y_{SB}^{\prime}\rho_{S}^{1/2}$, this means that%
\begin{align}
\max_{\rho_{S}}\min_{Y_{SB}\in\overrightarrow{\operatorname{SEP}}}\left\{
\operatorname{Tr}\{Y_{SB}\}:\rho_{S}^{1/2}J_{SB}^{\mathcal{N}}\rho_{S}%
^{1/2}\leq Y_{SB}\right\}   &  =\max_{\rho_{S}}\min_{Y_{SB}^{\prime}%
\in\overrightarrow{\operatorname{SEP}}}\left\{  \operatorname{Tr}\{\rho
_{S}Y_{SB}^{\prime}\}:J_{SB}^{\mathcal{N}}\leq Y_{SB}^{\prime}\right\}
\nonumber\\
&  =\min_{Y_{SB}^{\prime}\in\overrightarrow{\operatorname{SEP}}}\max_{\rho
_{S}}\left\{  \operatorname{Tr}\{\rho_{S}Y_{SB}^{\prime}\}:J_{SB}%
^{\mathcal{N}}\leq Y_{SB}^{\prime}\right\} \nonumber\\
&  =\min_{Y_{SB}^{\prime}\in\overrightarrow{\operatorname{SEP}}}\max_{\rho
_{S}}\left\{  \operatorname{Tr}\{\rho_{S}\operatorname{Tr}_{B}\{Y_{SB}%
^{\prime}\}\}:J_{SB}^{\mathcal{N}}\leq Y_{SB}^{\prime}\right\} \nonumber\\
&  =\min_{Y_{SB}^{\prime}\in\overrightarrow{\operatorname{SEP}}}\left\{
\left\Vert \operatorname{Tr}_{B}\{Y_{SB}^{\prime}\}\right\Vert _{\infty
}\}:J_{SB}^{\mathcal{N}}\leq Y_{SB}^{\prime}\right\} \nonumber\\
&  =\Sigma(\mathcal{N}).
\end{align}
The second equality follows from the Sion minimax theorem:\ the sets over
which we are optimizing are convex, with the set of density operators
additionally being compact, and the objective function $\operatorname{Tr}%
\{\rho_{S}Y_{SB}^{\prime}\}$ is linear in $\rho_{S}$ and $Y_{SB}^{\prime}$,
and so the Sion minimax theorem applies. The third equality follows from
partial trace, and the fourth follows because $\left\Vert D\right\Vert
_{\infty}=\max_{\rho}\operatorname{Tr}\{D\rho\}$, when the optimization is
with respect to density operators. Finally, we note that the ansatz may be
lifted by an appropriate limiting argument.
\end{proof}

\bigskip

We can now see that the expressions for $E_{\max}(A;B)_{\rho}$ in
Lemma~\ref{lem:alt-emax}\ and $E_{\max}(\mathcal{N})$ in
Lemma~\ref{lem:alt-e-max-channel} have a very similar form to those in
\eqref{eq:alt-max-Rains-state} and \eqref{eq:alt-max-Rains-channel} for
$R_{\max}(A;B)_{\rho}$ and $R_{\max}(\mathcal{N})$, respectively. However, the
optimization problems for $E_{\max}(A;B)_{\rho}$ and $E_{\max}(\mathcal{N})$
are not necessarily efficiently computable because they involve an
optimization over the cone of separable operators, which is known to be
difficult \cite{Harrow:2013} in general. Regardless, due to the forms that we
now have for $E_{\max}(A;B)_{\rho}$ and $E_{\max}(\mathcal{N})$, we can prove
an inequality from \cite{Christandl2017}, analogous to
\eqref{eq:r-max-amort-no-help-1}, with a proof very similar to that given in
the proof of Proposition~\ref{prop:amort-doesnt-help}:

\begin{proposition}
[\cite{Christandl2017}]\label{prop:amort-doesnt-help-e-max}Let $\rho
_{A^{\prime}AB^{\prime}}$ be a state and let $\mathcal{N}_{A\rightarrow B}$ be
a quantum channel. Then%
\begin{equation}
E_{\max}(A^{\prime};BB^{\prime})_{\omega}\leq E_{\max}(\mathcal{N})+E_{\max
}(A^{\prime}A;B^{\prime})_{\rho},
\end{equation}
where%
\begin{equation}
\omega_{A^{\prime}BB^{\prime}}=\mathcal{N}_{A\rightarrow B}(\rho_{A^{\prime
}AB^{\prime}}).
\end{equation}

\end{proposition}

\begin{proof}
By removing logarithms and applying Lemmas~\ref{lem:alt-emax} and
\ref{lem:alt-e-max-channel}, the desired inequality is equivalent to the
following one:%
\begin{equation}
W_{\operatorname{sep}}(A^{\prime};BB^{\prime})_{\omega}\leq\Sigma
(\mathcal{N})\cdot W_{\operatorname{sep}}(A^{\prime}A;B^{\prime})_{\rho},
\label{eq:no-logs-amort-no-help-e-max}%
\end{equation}
and so we aim to prove this one. Exploiting the identity in
Lemma~\ref{lem:alt-emax}, we find that%
\begin{equation}
W_{\operatorname{sep}}(A^{\prime}A;B^{\prime})_{\rho}=\min\operatorname{Tr}%
\{C_{A^{\prime}AB^{\prime}}\},
\end{equation}
subject to the constraints%
\begin{align}
C_{A^{\prime}AB^{\prime}}  &  \in\overrightarrow{\operatorname{SEP}}%
(A^{\prime}A\!:\!B^{\prime}),\\
C_{A^{\prime}AB^{\prime}}  &  \geq\rho_{A^{\prime}AB^{\prime}},
\label{eq:C-e-max-constr-2}%
\end{align}
while the identity in Lemma~\ref{lem:alt-e-max-channel} gives that%
\begin{equation}
\Sigma(\mathcal{N})=\min\left\Vert \operatorname{Tr}_{B}\{Y_{SB}\}\right\Vert
_{\infty},
\end{equation}
subject to the constraints%
\begin{align}
Y_{SB}  &  \in\overrightarrow{\operatorname{SEP}}(S\!:\!B),\\
Y_{SB}  &  \geq J_{SB}^{\mathcal{N}}. \label{eq:Y-e-max-chan-constr-2}%
\end{align}
The identity in Lemma~\ref{lem:alt-emax} implies that the left-hand side of
\eqref{eq:no-logs-amort-no-help-e-max}\ is equal to%
\begin{equation}
W_{\operatorname{sep}}(A^{\prime};BB^{\prime})_{\omega}=\min\operatorname{Tr}%
\{E_{A^{\prime}BB^{\prime}}\},
\end{equation}
subject to the constraints%
\begin{align}
E_{A^{\prime}BB^{\prime}}  &  \in\overrightarrow{\operatorname{SEP}}%
(A^{\prime}\!:\!BB^{\prime}),\label{eq:sep-constr-1}\\
E_{A^{\prime}BB^{\prime}}  &  \geq\mathcal{N}_{A\rightarrow B}(\rho
_{A^{\prime}AB^{\prime}}). \label{eq:sep-constr-2}%
\end{align}

With these optimizations in place, we can now establish the inequality in
\eqref{eq:no-logs-amort-no-help-e-max} by making a judicious choice for
$E_{A^{\prime}BB^{\prime}}$. Let $C_{A^{\prime}AB^{\prime}}$ be optimal for
$W_{\operatorname{sep}}(A^{\prime}A;B^{\prime})_{\rho}$, and let $Y_{SB}$ be
optimal for $\Sigma(\mathcal{N})$. Let $|\Upsilon\rangle_{SA}$ be the
maximally entangled vector, as defined in \eqref{eq:max-ent-vec}. Pick%
\[
E_{A^{\prime}BB^{\prime}}=\langle\Upsilon|_{SA}C_{A^{\prime}AB^{\prime}%
}\otimes Y_{SB}|\Upsilon\rangle_{SA}.
\]
This choice is clearly similar to that in the proof of
Proposition~\ref{prop:amort-doesnt-help}. We need to prove that $E_{A^{\prime
}BB^{\prime}}$ is feasible for $W_{\operatorname{sep}}(A^{\prime};BB^{\prime
})_{\omega}$. To this end, consider that%
\begin{align}
\langle\Upsilon|_{SA}C_{A^{\prime}AB^{\prime}}\otimes Y_{SB}|\Upsilon
\rangle_{SA}  &  \geq\langle\Upsilon|_{SA}\rho_{A^{\prime}AB^{\prime}}\otimes
J_{SB}^{\mathcal{N}}|\Upsilon\rangle_{SA}\nonumber\\
&  =\mathcal{N}_{A\rightarrow B}(\rho_{A^{\prime}AB^{\prime}}),
\end{align}
which follows from \eqref{eq:C-e-max-constr-2},
\eqref{eq:Y-e-max-chan-constr-2}, and \eqref{eq:Choi-iso}. Now, since
$C_{A^{\prime}AB^{\prime}}\in\overrightarrow{\operatorname{SEP}}(A^{\prime
}A\!:\!B^{\prime})$, it can be written as $\sum_{x}P_{A^{\prime}A}^{x}\otimes
Q_{B^{\prime}}^{x}$ for positive semi-definite $P_{A^{\prime}A}^{x}$ and
$Q_{B^{\prime}}^{x}$. Furthermore, consider that since $Y_{SB}\in
\overrightarrow{\operatorname{SEP}}(S\!:\!B)$, it can be written as $\sum
_{y}L_{S}^{y}\otimes M_{B}^{y}$ for positive semi-definite $L_{S}^{y}$ and
$M_{B}^{y}$. Then we have that%
\begin{align}
\langle\Upsilon|_{SA}C_{A^{\prime}AB^{\prime}}\otimes Y_{SB}|\Upsilon
\rangle_{SA}  &  =\sum_{x,y}\langle\Upsilon|_{SA}P_{A^{\prime}A}^{x}\otimes
Q_{B^{\prime}}^{x}\otimes L_{S}^{y}\otimes M_{B}^{y}|\Upsilon\rangle
_{SA}\nonumber\\
&  =\sum_{x,y}\langle\Upsilon|_{SA}P_{A^{\prime}A}^{x}T_{A}(L_{A}^{y})\otimes
Q_{B^{\prime}}^{x}\otimes I_{S}\otimes M_{B}^{y}|\Upsilon\rangle
_{SA}\nonumber\\
&  =\sum_{x,y}\operatorname{Tr}_{A}\{P_{A^{\prime}A}^{x}T_{A}(L_{A}%
^{y})\}\otimes Q_{B^{\prime}}^{x}\otimes M_{B}^{y}\nonumber\\
&  \in\operatorname{SEP}(A^{\prime}\!:\!BB^{\prime}).
\end{align}
The second equality follows from \eqref{eq:tr-trick} and the third from
\eqref{eq:max-ent-trace}. The last statement follows because
$\operatorname{Tr}_{A}\left\{  P_{A^{\prime}A}^{x}T_{A}(L_{A}^{y})\right\}
=\operatorname{Tr}_{A}\left\{  \sqrt{T_{A}(L_{A}^{y})}P_{A^{\prime}A}^{x}%
\sqrt{T_{A}(L_{A}^{y})}\right\}  $ is positive semi-definite for each $x$ and
$y$.\ Finally, consider that%
\begin{align}
\operatorname{Tr}\{E_{A^{\prime}BB^{\prime}}\}  &  =\operatorname{Tr}%
\{\langle\Upsilon|_{SA}C_{A^{\prime}AB^{\prime}}\otimes Y_{SB}|\Upsilon
\rangle_{SA}\}\nonumber\\
&  =\operatorname{Tr}\{C_{A^{\prime}AB^{\prime}}T_{A}(Y_{AB})\}\nonumber\\
&  =\operatorname{Tr}\{C_{A^{\prime}AB^{\prime}}T_{A}(\operatorname{Tr}%
_{B}\left\{  Y_{AB}\right\}  )\}\nonumber\\
&  \leq\operatorname{Tr}\{C_{A^{\prime}AB^{\prime}}\}\left\Vert T_{A}%
(\operatorname{Tr}_{B}\left\{  Y_{AB}\right\}  )\right\Vert _{\infty
}\nonumber\\
&  =\operatorname{Tr}\{C_{A^{\prime}AB^{\prime}}\}\left\Vert \operatorname{Tr}%
_{B}\left\{  Y_{AB}\right\}  \right\Vert _{\infty}\nonumber\\
&  =W_{\operatorname{sep}}(A^{\prime}A;B^{\prime})_{\rho}\cdot\Sigma
(\mathcal{N}).
\end{align}
The reasoning for this chain is identical to that for \eqref{eq:final-r-max-reasoning-amort-no-help}.

Thus, we can conclude that our choice of $E_{A^{\prime}BB^{\prime}}$ is
feasible for $W(A^{\prime};BB^{\prime})_{\omega}$. Since $W(A^{\prime
};BB^{\prime})_{\omega}$ involves a minimization over all $E_{A^{\prime
}BB^{\prime}}$ satisfying \eqref{eq:sep-constr-1} and \eqref{eq:sep-constr-2},
this concludes our proof of \eqref{eq:no-logs-amort-no-help-e-max}.
\end{proof}

\bigskip

By the same reasoning employed in the proof of
Corollary~\ref{cor:amort-not-enhance}, the equality in
\eqref{eq:max-rel-amort-no-help}\ follows as a consequence of the inequality
in Proposition~\ref{prop:amort-doesnt-help-e-max}.

We finally note that max-relative entropy of entanglement is subadditive as a
function of quantum channels, in the following sense:
\begin{equation}
E_{\max}(\mathcal{N} \otimes\mathcal{M}) \leq E_{\max}(\mathcal{N} ) +E_{\max
}( \mathcal{M})
\end{equation}
where $\mathcal{N}$ and $\mathcal{M}$ are quantum channels. This follows as a
consequence of the equality in \eqref{eq:max-rel-amort-no-help} and
\cite[Proposition~4]{KW17}, the latter of which states that the amortized
entanglement is always subadditive as a function of quantum channels. It is an
interesting open question to determine whether the max-relative entropy of
entanglement is additive as a function of quantum channels.

\section{On converses for quantum and private capacities}

\label{sec:discussion}Here we discuss briefly how our strong converse result
stands with respect to prior work on strong converses and quantum and private
capacities \cite{HW01,TWW14,MRW16,WD16,P17,WTB16,Christandl2017,WFD17}.

\subsection{Quantum capacities}

Let $Q(\mathcal{N})$ and $Q^{\dag}(\mathcal{N})$ denote the quantum capacity
and the strong converse quantum capacity of a quantum channel $\mathcal{N}$.
These quantities are defined similarly to $Q^{\operatorname{PPT-P}%
,\leftrightarrow}(\mathcal{N})$ and $Q^{\operatorname{PPT-P},\leftrightarrow
\dag}(\mathcal{N})$, but there is no PPT assistance allowed. The partial
transposition bound was defined in \cite{HW01} as follows:%
\begin{equation}
Q_{\Theta}(\mathcal{N})=\log_{2}\left\Vert T\circ\mathcal{N}\right\Vert
_{\Diamond},
\end{equation}
where $T$ denotes the transpose map and $\left\Vert \cdot\right\Vert
_{\Diamond}$ is the diamond norm. In \cite{HW01}, $Q_{\Theta}(\mathcal{N})$
was established as a pretty-strong converse rate, in the sense of \cite{MW13},
for the quantum capacity of the channel $\mathcal{N}$. This result was
subsequently improved in \cite{MRW16} to the following strong converse bound:
\begin{equation}
Q^{\dag}(\mathcal{N})\leq Q^{\operatorname{PPT-P},\leftrightarrow\dag
}(\mathcal{N})\leq Q_{\Theta}(\mathcal{N}).
\end{equation}
The recent work in \cite{WD16,WFD17}\ established the following two bounds:%
\begin{align}
R_{\max}(\mathcal{N})  &  \leq Q_{\Theta}(\mathcal{N}),\\
Q^{\dag}(\mathcal{N})  &  \leq R_{\max}(\mathcal{N}).
\end{align}
Thus, in light of the above history, it is clear that the natural question was
whether $Q^{\operatorname{PPT-P},\leftrightarrow\dag}(\mathcal{N})\leq
R_{\max}(\mathcal{N})$, and this is the question that our paper affirmatively
answers. In summary, we now have that%
\begin{equation}
Q(\mathcal{N})\leq Q^{\dag}(\mathcal{N})\leq Q^{\operatorname{PPT-P}%
,\leftrightarrow\dag}(\mathcal{N})\leq R_{\max}(\mathcal{N})\leq Q_{\Theta
}(\mathcal{N}).
\end{equation}

We now mention some other related results. The Rains relative entropy
$R(A;B)_{\rho}$\ of a bipartite state $\rho_{AB}$\ is defined as
\cite{R99,R01,AdMVW02}%
\begin{equation}
R(A;B)_{\rho}=\min_{\sigma_{AB}\in\text{PPT}^{\prime}(A:B)}D(\rho_{AB}%
\Vert\sigma_{AB}),
\end{equation}
where $D$ denotes the quantum relative entropy \cite{U62,Lindblad1973},
defined as $D(\omega\Vert\tau)=\operatorname{Tr}\{\omega\lbrack\log_{2}%
\omega-\log_{2}\tau]\}$ whenever $\operatorname{supp}(\omega)\subseteq
\operatorname{supp}(\tau)$ and $+\infty$ otherwise. Then the Rains information
$R(\mathcal{N})$ of a quantum channel $\mathcal{N}$ is defined by replacing
$E$ in \eqref{eq:channel-entang}\ with $R(A;B)_{\rho}$ \cite{TWW14}. One can
also define the amortized Rains information $R_{A}(\mathcal{N})$ via the
recipe in \eqref{eq:channel-entang-amort}. Due to the inequality
$D(\omega\Vert\tau)\leq D_{\max}(\omega\Vert\tau)$ \cite{D09}, the following
inequality holds%
\begin{equation}
R(\mathcal{N})\leq R_{\max}(\mathcal{N}).
\end{equation}
The following bound is known from \cite{TWW14}%
\begin{equation}
Q^{\dag}(\mathcal{N})\leq R(\mathcal{N}),
\end{equation}
and it is open to determine whether%
\begin{equation}
Q^{\operatorname{PPT-P},\leftrightarrow\dag}(\mathcal{N})\overset{?}{\leq
}R(\mathcal{N}).
\end{equation}
This latter inequality is known to hold if the channel $\mathcal{N}$ has
sufficient symmetry \cite{TWW14}.
%and approximately if the channel possesses
%some symmetry approximately \cite{KW17}.

The squashed entanglement $E_{\text{sq}}(A;B)_{\rho}$\ of a quantum state
$\rho_{AB}$ is defined as \cite{CW04}%
\begin{equation}
E_{\text{sq}}(A;B)_{\rho}=\frac{1}{2}\inf_{\rho_{ABE}}\left\{  I(A;B|E)_{\rho
}:\operatorname{Tr}_{E}\{\rho_{ABE}\}=\rho_{AB}\right\}  ,
\end{equation}
where $I(A;B|E)_{\rho}=H(AE)_{\rho}+H(BE)_{\rho}-H(E)_{\rho}-H(ABE)_{\rho}$
and $H(F)_{\sigma}=-\operatorname{Tr}\{\sigma_{F}\log_{2}\sigma_{F}\}$. (See
also discussions in \cite{T99,T02}\ for squashed entanglement.) One can also
consider the squashed entanglement of a channel $E_{\text{sq}}(\mathcal{N})$
\cite{TGW14IEEE}, as well as the amortized squashed entanglement
$E_{\text{sq,}A}(\mathcal{N})$. Another function of a quantum channel is its
entanglement cost \cite{BBCW13}, which we write as $E_{C}(\mathcal{N})$ and
for which a definition is given in \cite{BBCW13}. The following bounds and
relations are known regarding these quantities:%
\begin{align}
Q^{\leftrightarrow\dag}(\mathcal{N})  &  \leq E_{C}(\mathcal{N}), &
\cite{BBCW13}\\
E_{\text{sq},A}(\mathcal{N})  &  = E_{\text{sq}}(\mathcal{N}) , &
\cite{TGW14IEEE}\\
Q^{\leftrightarrow}(\mathcal{N})  &  \leq E_{\text{sq}}(\mathcal{N}) \leq
E_{C}(\mathcal{N}) , & \cite{TGW14IEEE}%
\end{align}
It is open to determine whether the following inequality holds%
\begin{equation}
Q^{\leftrightarrow\dag}(\mathcal{N})\overset{?}{\leq}E_{\text{sq}}%
(\mathcal{N}).
\end{equation}

\subsection{Private capacities}

One can also consider various private capacities and strong converse private
capacities of a quantum channel, denoted as $P(\mathcal{N})$,
$P^{\leftrightarrow}(\mathcal{N})$, $P^{\dag}(\mathcal{N})$, and
$P^{\leftrightarrow\dag}(\mathcal{N})$. Defining the relative entropy of
entanglement $E_{R}$ \cite{VP98}\ as%
\begin{equation}
E_{R}(A;B)_{\rho}=\min_{\sigma_{AB}\in\operatorname{SEP}(A:B)}D(\rho_{AB}%
\Vert\sigma_{AB}),
\end{equation}
and the max-relative entropy of entanglement $E_{\max}$\ as we did in
\eqref{eq:e-max}, we can also define their channel versions $E_{R}%
(\mathcal{N})$ and $E_{\max}(\mathcal{N})$\ and their amortized versions
$E_{R,A}(\mathcal{N})$ and $E_{\max,A}(\mathcal{N})$. For these various
quantities, we have that
\begin{align}
E_{R}(\mathcal{N})  &  \leq E_{\max}(\mathcal{N}), &  & \\
P^{\leftrightarrow}(\mathcal{N})  &  \leq E_{\text{sq}}(\mathcal{N}), &  &
\cite{TGW14IEEE,Wilde2016}\\
P^{\leftrightarrow\dag}(\mathcal{N})  &  \leq E_{C}(\mathcal{N}), &  &
\cite{Christandl2017}\\
E_{\max,A}(\mathcal{N})  &  =E_{\max}(\mathcal{N}), &  &
\cite{Christandl2017}\\
P^{\leftrightarrow\dag}(\mathcal{N})  &  \leq E_{\max}(\mathcal{N}), &  &
\cite{Christandl2017}\\
P^{\dag}(\mathcal{N})  &  \leq E_{R}(\mathcal{N}). &  &  \cite{WTB16}%
\end{align}
It is not known whether%
\begin{align}
P^{\leftrightarrow\dag}(\mathcal{N})  &  \overset{?}{\leq}E_{\text{sq}%
}(\mathcal{N}),\\
P^{\leftrightarrow\dag}(\mathcal{N})  &  \overset{?}{\leq}E_{R}(\mathcal{N}),
\end{align}
but the latter inequality is known to hold for channels with sufficient
symmetry \cite{WTB16}.
%, as well as for those having approximate symmetry
%\cite{KW17}.

An interesting question is whether the max-Rains information of a channel
$\mathcal{N}$ could serve as an upper bound on one of its private capacities
$P(\mathcal{N})$, $P^{\dag}(\mathcal{N})$, $P^{\leftrightarrow}(\mathcal{N})$,
or $P^{\leftrightarrow\dag}(\mathcal{N})$. The guiding principle behind many
strong converse bounds in quantum information theory is to compare the output
of the actual protocol, with respect to a relative-entropy-like measure, to a
state or positive semi-definite operator that is ``useless'' for the task. By
\textquotedblleft useless,\textquotedblright\ we mean that the state or
operator should have a probability of passing a test for the task that is no
larger than inversely proportional to the dimension of the system being
communicated. For example, this kind of result is known from \cite[Lemma~2]%
{R99}\ for operators in the set $\operatorname{PPT}^{\prime}(M_{A}\!:\!M_{B})$
and the entanglement test, and we used this bound effectively in
\eqref{eq:ent-test-PPT-states}\ in order to establish the max-Rains
information as an upper bound on PPT-P-assisted quantum capacity. Furthermore,
this kind of result is known from \cite{HHHO05,HHHO09,WTB16}\ for separable
states and the privacy test, and prior work has used this result to establish
upper bounds on various private capacities of a channel
\cite{WTB16,Christandl2017}. However, it is not known how to obtain this kind
of result for operators in the set $\operatorname{PPT}^{\prime}(M_{A}%
\!:\!M_{B})$ and the privacy test, and it is for this reason that we have not
been able to establish the max-Rains information as an upper bound on private
capacity. We doubt whether this would be possible, given that there exist
channels that produce PPT\ states with non-zero distillable secret key
\cite{HHHLO08,PhysRevLett.100.110502}.

In the same spirit, one might wonder about differences between the max-Rains
relative entropy and the max-relative entropy of entanglement. First, it is
clear that the max-relative entropy of entanglement can increase under the
action of a PPT-P channel, because there exist states that are PPT and
entangled \cite{HHHO05}. Furthermore, the aforementioned is related to the
fact that there exist states for which there is a strict separation between
the max-Rains relative entropy and the max-relative entropy of entanglement.
Any state that is PPT and entangled has a max-Rains relative entropy equal to
zero, while its max-relative entropy of entanglement is non-zero.

\subsection{Summary: Channel measures that do not increase under amortization}

In summary, we know that amortization does not increase

\begin{enumerate}
\item the squashed entanglement $E_{\text{sq}}(\mathcal{N})$ \cite{TGW14IEEE},

\item the max-relative entropy of entanglement $E_{\max}(\mathcal{N})$
\cite{Christandl2017},

\item or the max-Rains information $R_{\max}(\mathcal{N})$
(Corollary~\ref{cor:amort-not-enhance}).
\end{enumerate}

\noindent This is the main reason that these information quantities are
single-letter converse bounds for assisted capacities. Is there any chance
that the same could hold generally for $E_{R}(\mathcal{N})$ or $R(\mathcal{N}%
)$? If so, then the known capacity bounds could be improved.

\section{Conclusion}

\label{sec:conclusion}

The main contribution of our paper was to show that the max-Rains information
of a quantum channel does not increase under amortization. That is, when
entanglement is quantified by the max-Rains relative entropy, the net
entanglement that a channel can generate is the same as the amount of
entanglement that it can generate if the sender and receiver do not start with
any initial entanglement. This result then implies a single-letter,
strong-converse, and efficiently computable bound for the capacity of a
quantum channel to communicate qubits along with the assistance of
PPT-preserving operations between every channel use. As such, the max-Rains
information can be easily evaluated and is a general benchmark for this
capacity. As we emphasized previously, our upper bound is also an upper bound
on the physically relevant LOCC-assisted quantum capacity. The main tool that
we used to prove our result is the formulation of the max-Rains relative
entropy and max-Rains information as semi-definite programs
\cite{WD16,WD16pra,WFD17} (in particular, we employed semi-definite
programming duality---we note here that this kind of approach has previously
been employed successfully for multiplicativity, additivity, or parallel
repetition problems in quantum information theory
\cite{BT16,Berta2017,TCS-068}). We also compared our result to other results
in the growing literature on the topic of bounds for the assisted capacities
of arbitrary quantum channels
\cite{TGW14IEEE,TWW14,MRW16,WTB16,Christandl2017}.

We also provided an alternative proof for the fact that amortization does not
enhance a channel's max-relative entropy of entanglement \cite{Christandl2017}%
: i.e., $E_{\max,A}(\mathcal{N})=E_{\max}(\mathcal{N})$. This statement was
proved in \cite{Christandl2017}\ by employing the methods of complex
interpolation \cite{BL76}, but here we found a different proof by establishing
alternative expressions for the max-relative entropy of entanglement
(Lemma~\ref{lem:alt-emax}) and a channel's max-relative entropy of
entanglement (Lemma~\ref{lem:alt-e-max-channel}). These alternative
expressions then allowed us to employ reasoning similar to that in our proof
of Proposition~\ref{prop:amort-doesnt-help} in order to establish a different
proof for the equality $E_{\max,A}(\mathcal{N})=E_{\max}(\mathcal{N})$. We
suspect that our approach could be useful in future applications.

Finally, in \cite{WFD17}, it was noted that the max-Rains information does not
give a good upper bound on the quantum capacity of the qubit depolarizing
channel. Our result gives a compelling reason for this observation:\ the
max-Rains information finds its natural place as an upper bound on the
PPT-P-assisted quantum capacity of the qubit depolarizing channel, and these
assisting operations allowed between every channel use could result in a
significant increase in capacity.

\textit{Acknowledgements}.
We are grateful to Omar Fawzi, Xin Wang, David Reeb, Siddhartha Das, and
Andreas Winter for discussions related to the topic of this paper. We also thank the anonymous referee for comments that helped to improve our paper, in particular for comments about the usefulness of PPT-P channels. Part of
this work was done during the workshop \textquotedblleft Beyond I.I.D. in
Information Theory,\textquotedblright\ hosted by the Institute for
Mathematical Sciences, Singapore, 24-28 July 2017. MB acknowledges funding by
the SNSF through a fellowship. MMW acknowledges support from the Office of
Naval Research and the National Science Foundation under grant no.~1350397.

\bibliographystyle{alpha}
\bibliography{Ref}

\end{document}